\documentclass[letterpaper, 10pt, conference]{ieeeconf}

\IEEEoverridecommandlockouts

\usepackage[utf8]{inputenc}
\usepackage[english]{babel}

\usepackage{graphicx}
\graphicspath{{Figures/}}

\usepackage{newtxmath, newtxtext}

\newcommand{\alphabet}{\Sigma_{a}}

\newcommand{\AP}{\mathsf{AP}} 

\newcommand{\always}{\Box}
\newcommand{\eventually}{\Diamond}
\newcommand{\until}{\mathbin{\sf U}}

\newcommand{\nex}{\mathord{\bigcirc}}

\newcommand{\word}{\boldsymbol{\omega}}

\newcommand{\Lab}{\textsf{L}}   % Labeling of an output
\newcommand{\R}{\mathbb{R}}
\newcommand{\res}{\mathbf{r}}

\usepackage{mathtools, amssymb, amsthm, bm}

\newtheorem{lemma}{Lemma}

\theoremstyle{definition}
\newtheorem{definition}{Definition}

\theoremstyle{remark}
\newtheorem{remark}{Remark}

\let\labelindent\relax

\usepackage{array}
\usepackage{balance}
\usepackage{enumerate}
\usepackage{enumitem}
\usepackage{cite}
\usepackage{xcolor}
\usepackage{multirow}
\usepackage{booktabs}

\usepackage{hyperref}
\hypersetup{colorlinks = true, linkcolor = blue, citecolor = blue, urlcolor = blue}

\newcommand{\tbf}[1]{\textbf{#1}}
\newcommand{\mbf}[1]{\mathbf{#1}}
\DeclareMathOperator*{\U}{\mathcal{U}}
\DeclareMathOperator*{\W}{\mathcal{W}}
\newcommand{\wbar}{\overline{w}}

\allowdisplaybreaks

\title{\LARGE \bf Energetic Resilience under Temporal Logic Specifications}
\author{Ratnangshu Das$^*$, Ram Padmanabhan$^*$, Melkior Ornik, and Pushpak Jagtap% <-this % stops a space
\thanks{$^*$These authors contributed equally to this work.}
\thanks{Corresponding Author: Ram Padmanabhan. E-mail: ramp3@illinois.edu.}%
\thanks{Ratnagnshu Das and Pushpak Jagtap are with the Indian Institute of Science, Bengaluru KA, India.}%
\thanks{Ram Padmanabhan and Melkior Ornik are with the University of Illinois Urbana-Champaign, Urbana IL, USA.}%
\thanks{The work of Ratnangshu Das was supported by the Prime Minister’s Research Fellowship from the Ministry of Education, Government of India. The work of Pushpak Jagtap was supported by the AI \& Robotics Technology Park, the Google Research Grant and the SERB Startup Research Grant.}%
\thanks{The work of Ram Padmanabhan and Melkior Ornik was supported in part by the Air Force Office of Scientific Research under Grant FA9550-23-1-0131 and in part by the NASA University Leadership Initiative under Grant 80NSSC22M0070.}%
}

\begin{document}

\maketitle

\begin{abstract}
In environments with uncertainties or undesirable influences, control systems can require additional energy to achieve their task while remaining resilient to these influences. In this paper, we present an energetic resilience metric that quantifies the maximal additional energy used by a system under undesired effects, while satisfying complex specifications encoded through temporal logic. We prove that this metric satisfies properties that enable its computation even for compositions of these specifications, thus allowing considerations of sequential reachability and safety tasks. For specifications related to finite-horizon reachability and safety, we describe how synthesizing a control input and computing this metric reduces to solving efficient quadratic programs. Two case studies on a fighter-jet model and a planar mobile robot illustrate how the synthesized control inputs satisfy given specifications despite undesired and potentially adversarial effects. Further, we demonstrate how the energetic resilience metric varies with the initial state as well as the magnitude of undesired effects.
\end{abstract}

\section{Introduction} \label{sec:Introduction}
The quality of \emph{resilience} of a control system informally refers to its ability to continue achieving its specifications, despite the effect of undesirable influences from its environment. Such specifications include reachability, safety, optimality with respect to a specified criterion, or combinations of the above. Undesired effects seek to prevent the achievement of these objectives, and can include simple bounded perturbations \cite{Yan23}, changes in input structure \cite{BO23} or even abrupt changes in dynamics through system faults \cite{AH19}. Resilience is a key factor in a variety of practical areas including infrastructure \cite{rehak19}, cyber-physical systems and networks \cite{PXG20, ZZ24}, and aircraft and spacecraft \cite{HLXZ21}. 

Linear temporal logic (LTL) \cite{BK08} is often used to formally encode complex specifications such as safety and set reachability for control systems. Under these specifications, an important task is to provide a quantitative measure of how resilient a given system is. The objective of this paper is to introduce a resilience metric which seeks to quantify the additional energy used by inputs of a control system to achieve LTL specifications, under undesirable influences.

Existing work in quantifying resilience under temporal logic specifications has generally involved computing robustness or uncertainty thresholds under which the specifications continue to be satisfied. For instance, \cite{FMPS19} and \cite{IHL23} use the notion of robustness for stochastic systems with respect to signal temporal logic (STL) specifications for resilience quantification and control synthesis. Along similar lines, satisfying specifications in a robust manner under a given set of disturbances has also been studied \cite{Liu24, Zhang24}. The works of Saoud \emph{et al.} \cite{SJS25, monir2025computation} define resilience as the maximum disturbance that can be tolerated while still satisfying a given temporal logic specification. In \cite{ait2025maximally}, the authors further address the synthesis of controllers that maximize this notion of resilience. However, these works do not quantify the additional control effort required to satisfy the specification under such disturbances.

The notion of control energy is an especially useful one for practical systems. The amount of fuel expended by a vehicular system can be closely related to energy in the control signal \cite{Na20}. Designing controllers for fuel economy is a well-studied topic, including studies in autonomous vehicles \cite{MS16}, aircraft \cite{Elias22} and hybrid vehicles \cite{HV1}. Quantifying resilience through considerations of control energy generally involves comparing the energy used by all inputs of the system under undesirable effects to the energy used under nominal behavior. Such a comparison is performed in \cite{PBDO24} under simple disturbances with no input constraints, as well as in \cite{PO25a, PO26a} under actuator malfunctions. However, these works consider very simple specifications of reaching a specified target state in finite-time, and do not involve safety or more complex sequential reachability specifications.

Under these considerations, in this paper, we present a resilience metric that quantifies the maximal additional energy used by the inputs of a control system under undesirable influence, compared to its nominal behavior, under temporal logic specifications. This metric is termed \emph{energetic resilience}, and we consider undesirable influences that can be seen either as an exogenous disturbance or uncontrolled, adversarial inputs. We present the notions of control energy under both nominal and malfunctioning behavior, under constraints on both the input and the exogenous quantity. We derive basic compositional properties of the energetic resilience metric, and for certain classes of specifications, we demonstrate how the energetic resilience metric can be computed efficiently through convex programs. We illustrate the applicability of this framework on two case studies, a fighter-jet model and a planar mobile robot. These examples demonstrate how complex tasks in practical systems are achieved by control inputs synthesized from the convex programs. We also show how the energetic resilience varies with initial state as well as magnitude of undesired effects for the planar mobile robot satisfying a reachability specification. In contrast to \cite{PO25a, PO26a}, which consider simpler task structures, we address significantly more complex temporal logic specifications and provide computational methods for evaluating energetic resilience. Moreover, unlike existing resilience metrics \cite{SJS23, SJS25, IHL23, Zhang24}, our formulation explicitly captures the role of control energy. Notably, \cite{SJS23, SJS25} focus on systems without control inputs, whereas we quantify the additional control effort required to achieve task satisfaction under disturbances.

The remainder of this paper is organized as follows. In Section~\ref{sec:Preliminaries}, we setup the problem, providing background on LTL specifications. In Section~\ref{sec:Metrics}, we introduce the energetic resilience metric and present some relevant properties of this metric in the context of LTL specifications. In Section~\ref{sec:Computation}, we describe how this metric can be computed using convex programs for different classes of LTL specifications, as well as compositions of such specifications. In Section~\ref{sec:Examples}, we present two illustrative examples that compute the additional energy used under undesirable effects, to achieve complex specifications. We also show how the energetic resilience metric varies with the initial state and magnitude of undesired effects, and how this metric is used to characterize the additional energy required for a given specification.

\tbf{Notation:} The $p$-norm of a vector $x \in \mathbb{R}^n$ is given by $\|x\|_p \coloneqq \left(\sum_{i = 1}^{n} |x_i|^p\right)^{1/p}$, with $\|x\|_{\infty} \coloneqq \sup_i |x_i|$. For a sequence $\{z_t\}$ with each $z_t \in \mathbb{R}^n$, the $\ell_2$-norm is defined by $\|z\|_{\ell_2} \coloneqq \sqrt{\sum_{t} \|z_t\|_{2}^{2}}$. For two matrices $A$ and $B$, $A\otimes B$ denotes the Kronecker product between these matrices. The vector $\mbf{1}_n$ denotes the vector of ones in $\mathbb{R}^n$ and the matrix $I_n \in \mathbb{R}^n$ denotes the identity matrix of dimension $n$.

\section{Preliminaries} \label{sec:Preliminaries}
%\tb{Includes system description, LTL basics, defining relevant energies (nominal/malfunctioning) and energetic resilience.}

We consider discrete-time systems evolving on a state space $X \subseteq \mathbb{R}^n$, over a finite-horizon indexed by $t = 0, \ldots, N$, according to the dynamics
\begin{equation} \label{eq:System}
x_{t+1} = Ax_t + B_uu_t + B_ww_t,
\end{equation}
where $x_t \in X$ is the state, $u_t \in \mathbb{R}^m$ is the control input, and $w_t \in \mathbb{R}^p$ is an exogenous, uncontrolled quantity that can be seen either as a disturbance or an adversarial input. We assume $u_t$ and $w_t$ are constrained to lie in convex polyhedral sets, or \emph{polytopes} \cite{Ziegler95} denoted $\U$ and $\W$, which are represented as
\begin{align}
	\U &\coloneqq \{u_t \in \mathbb{R}^m: H_{\U} u_t \leq h_{\U}\}, \label{eq:SetU} \\
	\W &\coloneqq \{w_t \in \mathbb{R}^p: \|w_t\|_{\infty} \leq \wbar\}. \label{eq:SetW}
\end{align}
We assume the pair $(A,B_u)$ is controllable. The specifications we consider throughout this paper are described through the framework of linear temporal logic (LTL) \cite{BK08}.

\subsection{Linear Temporal Logic Specifications} \label{sec:LTLF}

Linear Temporal Logic (LTL) is a formal language used to describe how a system should behave over time. It is built using a set of atomic propositions and combined using logical and temporal operators \cite{BK08}.

Let $\AP$ be a finite set of atomic propositions, and define the alphabet as $\alphabet := \AP$. Each letter $\alpha \in \alphabet$ evaluates a subset of the atomic propositions as true.

In this work, we consider LTL over finite sequences, known as LTL$_F$ \cite{de2013linear}. A finite word of length $N \in \mathbb{N}$ is written as
$\word = (\word_0, \word_1, \ldots, \word_{N-1}) \in \alphabet^N,$
where $\word_i \in \alphabet$ represents the letter of the word at $i$th instance.

We connect system trajectories to words using a measurable labeling function $\Lab : X \rightarrow \alphabet$, which assigns letters $\alpha =\Lab(x)$ to state $x\in X$. Any finite trajectory $w_x = (x_0,x_1,\ldots,x_{N-1})$ is mapped to the set of finite words $\alphabet^{N}$, as
$\word = \Lab(w_x) := (\Lab(x_0), \Lab(x_1), \ldots, \Lab(x_{N-1})).$

\begin{definition}
\label{def:LTL}
An LTL$_F$ formula over a set of atomic propositions $\AP$ is constructed recursively as
\begin{equation*}
	% \label{eq:PNF}
\psi ::=  \textsf{true}  |  p  |  \neg \psi
     | \psi_1 \wedge \psi_2  |  \psi_1 \vee \psi_2  |  \nex \psi  |  \psi_1\until \psi_2  |  \square\psi | \lozenge\psi,
\end{equation*}
where $p \in \AP$ and $\psi, \psi_1, \psi_2$ are LTL$_F$ formulas.
\end{definition}

Given a finite word $\word$ of length $N$ and an LTL$_F$ formula $\psi$, we define when $\psi$ is satisfied at the $i^{\text{th}}$ step $(i < N)$ (denoted $\word_{i}\models\psi$) as follows:
\begin{itemize}
\item $\word_i \models \textsf{true}$ always holds, and $\word_i \models \textsf{false}$ never holds.
\item $ \word_i\vDash p$  for $ p\in \AP$ holds if $p \in\word_{i}$.
\item A negation, $\word_i\vDash\neg p$, holds if $ \word_i\nvDash p$.
\item A logical conjunction, $\word_i\vDash \psi_1\wedge\psi_2$, holds
if $ \word_i\vDash \psi_1$ and $ \word_i\vDash \psi_2$.
\item A logical disjunction, $\word_i\vDash \psi_1\vee\psi_2$, holds
if $ \word_i\vDash \psi_1$ or $ \word_i\vDash \psi_2$.
\item A temporal next operator, $\word_i\vDash\nex\psi$, holds if $\word_{i+1}\vDash \psi$. Similarly, for $0 \leq j < N-i$, $\word_i\vDash\nex^j\psi$, holds if $\word_{i+j}\vDash \psi$. 
\item A temporal until operator, $\word_i\vDash \psi_1\until\psi_2$, holds if for some $m$ such that $i\leq m< N$, we have $\word_m\vDash\psi_2$ and for all $i\leq k< m$, we have $\word_k\vDash\psi_1$.
\item A temporal always operator, $\word_i\vDash \square\psi$, holds if for all $m$ such that $i\leq m< N$, we have $\word_m\vDash \psi$. Similarly, for $0 \leq j < N-i$, $\word_i\vDash \square^j\psi$ holds if for all $i \leq m \leq i+j$, we have $\word_m\vDash \psi$.
\item A temporal eventually operator, $\word_i\vDash \lozenge\psi$, holds if for some $m$ such that $i\leq m< N$, we have $\word_m\vDash \psi$. Similarly, for $0 \leq j < N-i$, $\word_i\vDash \lozenge^j\psi$, holds if for some $m$ such that $i\leq m< i+j$, we have $\word_{m}\vDash \psi$. 
\end{itemize}
An LTL$_F$ formula $\psi$ is true for $\omega$, $w \models \psi$, iff $\omega_0 \models \psi$. For a trajectory $\omega_x = (x_0, \ldots, x_{N-1}) \in X^N$, we say that $\omega_x \models \psi$ if for $\omega = \Lab(\omega_x) := (\Lab(x_0), \ldots, \Lab(x_{N-1}))$, , we have $\omega \models \psi$. Similarly, for a set of trajectories $\mathcal{X} \subseteq X^N$, we say that $\mathcal{X} \models \psi$, if $\omega_x \models \psi$ for all $\omega_x \in \mathcal{X}.$

\section{Metrics for Energy Quantification} \label{sec:Metrics}
In this section, we introduce metrics for resilience quantification based on the energy used by the control input $u_t$. Let $u$ be the sequence $\{u_t\}$ for $t = 0, \ldots, N-1$, and $\xi_u(x_0, \W)$ be the set of trajectories of any length from an initial state $x_0$, under a given input sequence $u$, and $w_t \in \W$ in \eqref{eq:System}, i.e., 
\begin{align}
\xi_u(x_0, \W) = \big\{(x_0, x_1, x_2, \ldots) : x_{t+1} &= Ax_t + B_uu_t + B_ww_t, \nonumber \\
u = (u_0, u_1, \ldots), \quad w_t &\in \W\big\}. \label{eq:Trajectories}
\end{align}
Clearly, $\xi_u(x_0, \{0\})$ denotes the set of trajectories under \emph{nominal} dynamics of the form
\begin{equation} \label{eq:Nominal}
	x_{t+1} = Ax_t + B_uu_t.
\end{equation}
Given an LTL$_F$ specification $\psi$, and an initial state $x_0$, we define the {nominal energy} as follows.

\begin{definition}[Nominal Energy]
The \emph{nominal energy} $E_{nom}(x_0; \psi)$ is defined as the minimum energy in the control input used by the dynamics \eqref{eq:Nominal} to achieve an LTL$_F$ specification $\psi$ from a given initial state $x_0$, i.e.,
\begin{equation} \label{eq:EN_def}
	E_{nom}(x_0; \psi) \coloneqq \inf_{u} \left\{\|u\|_{\ell_2}^{2} \text{ s.t. } \xi_u(x_0, \{0\}) \models \psi, u_t \in \U\right\}.
\end{equation}
\end{definition}

Similarly, we define the {malfunctioning energy} as follows.

\begin{definition}[Malfunctioning Energy]
The \emph{malfunctioning energy} $E_{mal}(x_0; \psi)$ is defined as the minimum energy in the control input used by the dynamics \eqref{eq:System} to achieve a specification $\psi$ from a given initial state $x_0$, when $w_t$ is chosen to maximize this quantity, i.e.,
\begin{equation} \label{eq:EM_def}
	E_{mal}(x_0; \psi)\!\coloneqq\!\sup_{w \in \W} \left\{\inf_{u} \left\{\|u\|_{\ell_2}^{2} \text{ s.t. } \xi_u(x_0, \W) \models \psi, u_t \in \U\right\}\right\}.
\end{equation}
\end{definition}
The distinguishing feature between these energies is the underlying dynamics that are used to formulate them. Moreover, the quantity $w_t$ seeks to maximize the malfunctioning energy. We assume the specification $\psi$ is \emph{non-trivial}, in the sense that satisfying $\psi$ requires nonzero energy in the control input $u_t$ under either dynamics \eqref{eq:System} or \eqref{eq:Nominal}. Using these energies, we now define an \emph{energetic resilience} metric for a given initial state $x_0$ under a specification $\psi$.

\begin{definition}[Energetic Resilience]
The \emph{energetic resilience} $\res(x_0; \psi)$ of a given initial state $x_0$ under the LTL$_F$ specification $\psi$ is the difference between the malfunctioning and nominal energies, i.e., 
\begin{equation} \label{eq:rA_def}
\res(x_0; \psi) \coloneqq E_{mal}(x_0; \psi) - E_{nom}(x_0; \psi).
\end{equation}
\end{definition}

Clearly, for a given initial state $x_0$ and specification $\psi$, $\res(x_0;\psi)$ quantifies the additional energy required to satisfy $\psi$. We remark that a similar metric was posed in prior work \cite{PO25a, PO26a} but under significantly simpler specifications, which only required achieving a target state $x_{tg}$ in finite time. Such simplicity in specifications allowed deriving analytical bounds on the different energies defined above, as well as the energetic resilience metric. In contrast, we consider more complex specifications encoded through $\psi$ in this paper. As we show in the next section, obtaining the energies defined above as well as the energetic resilience metric reduces to solving a set of efficient convex programs, allowing us to characterize the maximal energy required to execute more complex tasks. We now present some key properties of the metric \eqref{eq:rA_def}, which can be used to compute the metric while satisfying compositions of more elementary temporal logic specifications.

\begin{figure*}[t]
    \centering
    \includegraphics[width=0.85\linewidth]{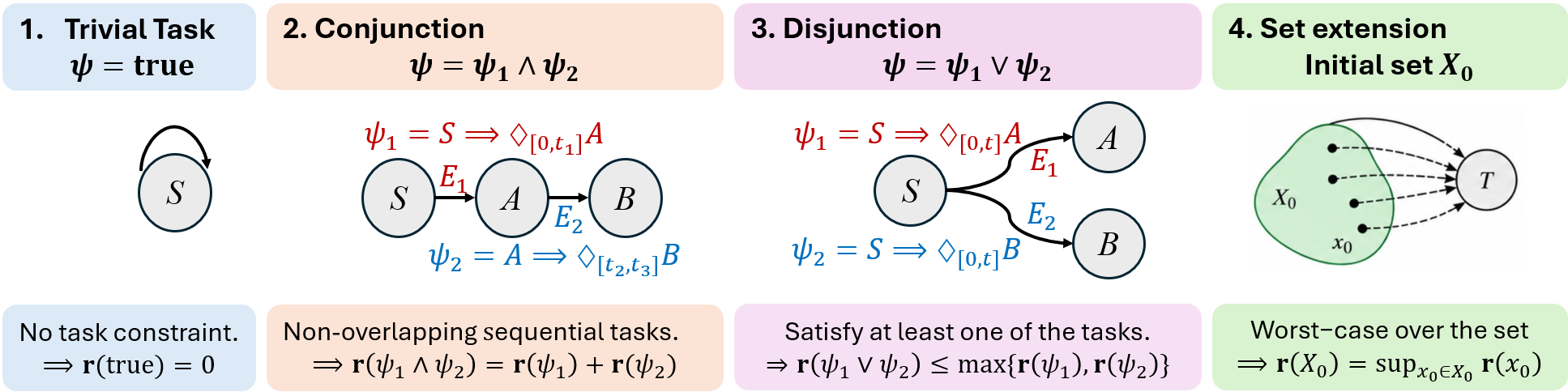}
    \caption{Intuitive illustration of the compositional properties of energetic resilience.}
    \label{fig:resilience_composition}
    \vspace{-0.5cm}
\end{figure*}

\begin{lemma}[Compositional Properties of Energetic Resilience]
\label{lem:energy_composition_full}
Consider the system \eqref{eq:System}, an initial state $x_0 \in X$, and an $\mathrm{LTL}_F$ specification $\psi$.
Let $E_{nom}^*(x_0;\psi)$ and $E_{mal}^*(x_0;\psi)$ denote the optimal nominal and malfunctioning control efforts defined in \eqref{eq:EN_def}-\eqref{eq:EM_def}. Then, for specifications $\psi_1$ and $\psi_2$ with non-overlapping task horizons:
\begin{enumerate}
    \item Trivial specification: For $\psi = \mathsf{true}$,
    $\res(x_0; \psi) = 0$.
    
    \item Conjunction: 
    $\res(x_0;\psi_1 \wedge \psi_2) = (\res(x_0;\psi_1) + \res(x_1;\psi_2))$.
    
    \item Disjunction: 
    $\res(x_0;\psi_1 \!\vee\! \psi_2) \!\leq\! \max \{\res(x_0;\psi_1), \!\res(x_0;\psi_2)\}$.

    \item Set extension: 
    For $X_0 \!\subseteq\! X$,
    $\res(X_0;\psi) \leq \sup_{x_0 \in X_0} \res(x_0;\psi)$.
\end{enumerate}
\end{lemma}

\begin{proof}
The proof follows from the definition of energetic resilience in Equation~\eqref{eq:rA_def} together with the compositional properties of the nominal and malfunctioning energies.

\medskip

\noindent 1)
If $\psi = \mathsf{true}$, then the specification is satisfied by any trajectory under both nominal and malfunctioning dynamics. Hence,
$$
E_{nom}^*(x_0;\psi) = E_{mal}^*(x_0;\psi) = 0,
\implies
\res(x_0;\psi) = 0.
$$
We remark that this statement implies only that the metric defined in \eqref{eq:rA_def} equals zero, and does not imply the system is \emph{not resilient}.

\medskip

\noindent 2)
Since the task horizons of $\psi_1$ and $\psi_2$ do not overlap, we may assume that the nominal and malfunctioning control efforts for the conjunction are additively separable across the two sub-tasks. That is,
\begin{align*}
E_{nom}^*(x_0;\psi_1\wedge\psi_2) &= E_{nom}^*(x_0;\psi_1) + E_{nom}^*(x_1;\psi_2),\\
E_{mal}^*(x_0;\psi_1\wedge\psi_2) &= E_{mal}^*(x_0;\psi_1) + E_{mal}^*(x_1;\psi_2).
\end{align*}
Consequently,
\begin{align*}
&\res(x_0; \psi_1 \wedge \psi_2) = E_{mal}^*(x_0;\psi_1 \wedge \psi_2) - E_{nom}^*(x_0;\psi_1 \wedge \psi_2) \\
&= \big(E_{mal}^*(x_0;\psi_1) + E_{mal}^*(x_1;\psi_2)\big) \\
&\qquad\qquad - \big(E_{nom}^*(x_0;\psi_1) + E_{nom}^*(x_1;\psi_2)\big) \\
&= \big(E_{mal}^*(x_0;\psi_1) - E_{nom}^*(x_0;\psi_1)\big)\\
&\qquad\qquad+ \big(E_{mal}^*(x_1;\psi_2) - E_{nom}^*(x_1;\psi_2)\big)\\
&=\res(x_0; \psi_1) + \res(x_0; \psi_2).
\end{align*}

\medskip

\noindent 3) Disjunction:
Let $\psi = \psi_1 \vee \psi_2$. Since satisfying $\psi$ requires satisfying at least one of the two sub-specifications, the feasible set of control inputs for $\psi$ is the union of the feasible control sets for $\psi_1$ and $\psi_2$. Hence,
\begin{align*}
E_{mal}^*(x_0;\psi_1 \vee \psi_2) &= \min\{E_{mal}^*(x_0;\psi_1),      E_{mal}^*(x_0;\psi_2)\},\\
E_{nom}^*(x_0;\psi_1 \vee \psi_2) &= \min\{E_{nom}^*(x_0;\psi_1),      E_{nom}^*(x_0;\psi_2)\}.
\end{align*}
Therefore,
\begin{align*}
\res(x_0;\psi_1\vee\psi_2) &= \min_i E_{mal}^*(x_0;\psi_i) - \min_i E_{nom}^*(x_0;\psi_i)\\
&\leq \max_i \left|E_{mal}^*(x_0;\psi_i) - E_{nom}^*(x_0;\psi_i)\right|\\
&= \max\{ \res(x_0;\psi_1),\res(x_0;\psi_2)\},
\end{align*}
where the last equality follows from
$E_{mal}^*(x_0;\psi_i)\geq E_{nom}^*(x_0;\psi_i)$.

\medskip

\noindent {4) Set extension:}
Let $X_0 \subseteq X$. By definition, the energetic resilience over a set is obtained by aggregating the pointwise resilience values over all $x_0 \in X_0$. Since every pointwise value satisfies
$\res(x_0;\psi) \le \sup_{x_0 \in X_0} \res(x_0;\psi),$
it follows directly that
$$
\res(X_0;\psi) \le \sup_{x_0 \in X_0} \res(x_0;\psi).
$$

This completes the proof. The corresponding intuition for the four
compositional cases is illustrated in Figure~\ref{fig:resilience_composition}.
\end{proof}

\begin{remark}
The above results provide upper bounds on energetic resilience under compositions. In general, exact expressions for composed tasks are difficult to obtain due to coupling between control inputs across sub-tasks. 
% For example, the energy associated with conjunction can be written as
% $$E(x_0; \psi_1 \wedge \psi_2) = \inf_{u} \left\{ \|u\|^2 \;\middle|\; u \text{ satisfies both } \psi_1 \text{ and } \psi_2 \right\},$$
% which highlights the inherent coupling in the optimization problem. 
Deriving tighter bounds or exact metrics for composed specifications, potentially using decomposition techniques, is an interesting direction for future work.
\end{remark}

\section{Computation of Energetic Resilience} \label{sec:Computation}
%\tb{Different convex programs for different specifications like safety, reachability. Discussing how to chain specifications together, and discussing how to evaluate energetic resilience given an initial condition. I can work on this section.}

In this section, we explain how energetic resilience metrics can be computed exactly for certain classes of specifications, such as exact-time reachability, finite-horizon safety, and finite-horizon reachability of convex polytopes $\Gamma$:
\begin{itemize}
    \item $\psi = \nex^N \Gamma$: Exact-time reachability requires the system to reach a target set $\Gamma$ at a specific time $N$.

    \item $\psi = \lozenge^N \Gamma$: Finite-horizon reachability requires the system to reach $\Gamma$ at some time on or before $N$.

    \item $\psi = \square^N \Gamma$: Finite-horizon safety requires the system to remain within the set $\Gamma$ at all times up to time $N$.
\end{itemize}
While similar specifications are considered in \cite{SJS25}, that work does not address systems with control inputs. Unlike \eqref{eq:rA_def}, there is no notion of energy in the control signal in \cite{SJS25}, and indeed, the resilience of a system was defined only through the disturbance. In contrast, our notion of resilience is based on additional energy used in the control signal, resulting in different objectives in the resulting computations.

\subsection{Exact-Time Reachability}
We discuss how the nominal \eqref{eq:EN_def} and malfunctioning \eqref{eq:EM_def} energies are computed for the specification $\psi = \nex^N \Gamma$, i.e., ensuring $x_N \in \Gamma$ for a specific time $N$. Let the polytope $\Gamma \coloneqq \{y \in \mathbb{R}^n: H_{\Gamma}y \leq h_{\Gamma}\}$ where $H_{\Gamma} \in \mathbb{R}^{n\times n}$ and $h_{\Gamma} \in \mathbb{R}^n$. In the nominal dynamics \eqref{eq:Nominal}, we know
\begin{equation} \label{eq:xN_nom}
x_N = A^Nx_0 + A^{N-1}B_uu_0 + \ldots + AB_uu_{N-2} + B_uu_{N-1},
\end{equation}
and we must have $H_{\Gamma}x_N \leq h_{\Gamma}$. Let $\mbf{F}_u \coloneqq [H_{\Gamma}A^{N-1}B_u ~~ \ldots ~~ H_{\Gamma}AB_u ~~ H_{\Gamma}B_u] \in \mathbb{R}^{n\times mN}$ and $\mbf{u} = [u_{0}^{\top}, \ldots, u_{N-1}^{\top}]^{\top} \in \mathbb{R}^{mN}.$ Then, the requirement is
\begin{equation} \label{eq:xN_1}
	H_{\Gamma}A^Nx_0 + \mbf{F}_u \mbf{u} \leq h_{\Gamma}
\end{equation}
using \eqref{eq:xN_nom}. We also have the constraint $u_t \in \U$ for all $t$ \eqref{eq:SetU}. Then, the nominal energy \eqref{eq:EN_def} is given by the quadratic program:
\begin{align}
E_{nom}(x_0; \psi) = &\min_{\mbf{u}} \|\mbf{u}\|_{2}^{2} \nonumber \\
\text{s.t. }~ &\mbf{F}_u\mbf{u} \leq h_{\Gamma} - H_{\Gamma}A^Nx_0, \nonumber \\
&H_{\U}u_t \leq h_{\U} \text{ for all $t$}. \label{eq:EN_exreach}
\end{align}
Similarly, under the original dynamics \eqref{eq:System}, the state at time $N$ is given by
\begin{align}
x_N &= A^Nx_0 + A^{N-1}B_uu_0 + \ldots + AB_uu_{N-2} + B_uu_{N-1} \nonumber \\
&+ A^{N-1}B_ww_0 + \ldots + AB_ww_{N-2} + B_ww_{N-1}. \label{eq:xN_malf}
\end{align}
Then, the requirement $H_{\Gamma}x_N \leq h_{\Gamma}$ can be rewritten as
\begin{equation} \label{eq:xN_malf1}
	H_{\Gamma}A^Nx_0 + \mbf{F}_u \mbf{u} + \mbf{F}_w \mbf{w} \leq h_{\Gamma},
\end{equation}
where $\mbf{F}_w \coloneqq [H_{\Gamma}A^{N-1}B_w ~~ \ldots ~~ H_{\Gamma}AB_w ~~ H_{\Gamma}B_w] \in \mathbb{R}^{n\times mN} \in \mathbb{R}^{n\times pN}$ and $\mbf{w} = [w_{0}^{\top}, \ldots, w_{N-1}^{\top}]^{\top} \in \mathbb{R}^{pN}$. Since $\mbf{w}$ seeks to maximize the malfunctioning energy \eqref{eq:EM_def} and the above equation is linear in both $\mbf{u}$ and $\mbf{w}$, we replace $\mbf{F}_w \mbf{w}$ by its worst-case effect $\|\mbf{F}_w\|_1\wbar \mbf{1}_n$ in \eqref{eq:xN_malf1} to obtain the constraint
\begin{equation} \label{eq:xN_malf2}
	H_{\Gamma}A^Nx_0 + \mbf{F}_u \mbf{u} + \|\mbf{F}_w\|_1 \wbar \mbf{1}_n \leq h_{\Gamma}.
\end{equation}
Then, the malfunctioning energy is given by the following quadratic program:
\begin{align}
E_{mal}(x_0; \psi) = &\min_{\mbf{u}} \|\mbf{u}\|_{2}^{2} \nonumber \\
\text{s.t. }~ &\mbf{F}_u \mbf{u} \leq h_{\Gamma} - H_{\Gamma}A^Nx_0 - \|\mbf{F}_w\|_1 \wbar \mbf{1}_n, \nonumber \\
&H_{\U}u_t \leq h_{\U} \text{ for all $t$}. \label{eq:EM_exreach}
\end{align}
Finally, the energetic resilience metric \eqref{eq:rA_def} for exact-time reachability can be obtained using the difference of the two quadratic programs \eqref{eq:EN_exreach} and \eqref{eq:EM_exreach}.

\subsection{Finite-horizon Reachability}
The finite-horizon reachability specification $\psi = \Diamond^N \Gamma$ can be equivalently written as a disjunction between exact-time reachability specifications $\nex^t \Gamma$, i.e.,
$$
\psi = \bigvee_{t=1}^{N} \nex^t \Gamma.
$$
Thus, the control objective reduces to reaching $\Gamma$ at \emph{some} time $t \in \{1,\ldots,N\}$. For each $t$, we can compute the nominal and malfunctioning energies required to satisfy $\nex^t \Gamma$ using \eqref{eq:EN_exreach} and \eqref{eq:EM_exreach}, and we denote these as $E_{nom}^{t}(x_0)$ and $E_{mal}^{t}(x_0)$, with a slight abuse of notation. Since satisfying $\psi$ requires satisfying at least one of the sub-specifications $\nex^t \Gamma$, the nominal and malfunctioning energies \eqref{eq:EM_def} and \eqref{eq:EN_def} are obtained by selecting the most favorable time index:
\begin{equation} \label{eq:ENEM_fhreach}
E_{nom}(x_0; \psi) = \min_{t} E_{nom}^{t}(x_0), ~~
E_{mal}(x_0; \psi) = \min_{t} E_{mal}^{t}(x_0).
\end{equation}
Subsequently, the energetic resilience \eqref{eq:rA_def} is computed using the difference of these quantities. We note that the reach time in the nominal and malfunctioning cases need not be the same, since the specification $\psi = \Diamond^N \Gamma$ only specifies an upper limit $N$ on this reach time.

\begin{figure*}[t]
    \centering
    \includegraphics[width=0.9\linewidth]{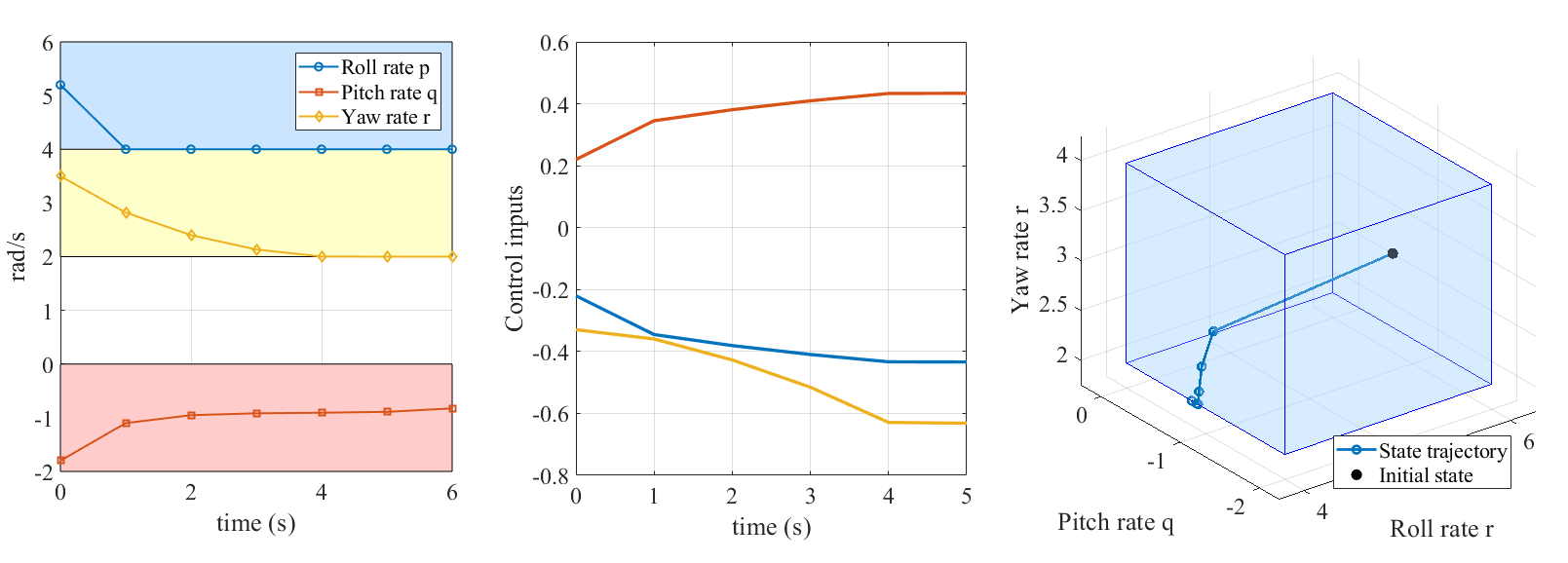}
    \caption{ADMIRE system under the finite-horizon safety specification $\always^6(x_k \in \mathcal{S})$. \emph{Left:} Evolution of the system states over time, showing that the trajectory remains within the prescribed safe bounds throughout the entire horizon. \emph{Middle:} Corresponding control inputs applied to enforce safety. \emph{Right:} Three-dimensional state evolution of the system, illustrating that the trajectory remains inside the safe region over the finite horizon.}
    \label{fig:admire_safety}
\end{figure*}

\subsection{Finite-horizon Safety}
This specification can be written as $x_t \in \Gamma$ for all $t = 1, \ldots, N$, and we implicitly assume $x_0 \in \Gamma$ is satisfied. As before, we consider the polytope $\Gamma = \{y \in \mathbb{R}^n: H_{\Gamma}y \leq h_{\Gamma}\}$. Let $\mbf{x} = [x_{1}^{\top}, \ldots, x_{N}^{\top}]^{\top} \in \mathbb{R}^{nN}$, $\mbf{H_{\Gamma}} = I_N \otimes H_{\Gamma} \in \mathbb{R}^{nN\times nN}$ and $\mbf{h_{\Gamma}} = \mbf{1}_N \otimes h_{\Gamma} \in \mathbb{R}^{nN}$. Then, $\psi = \square^N \Gamma$ can be rewritten as
$$
\mbf{H_{\Gamma}x} \leq \mbf{h_{\Gamma}}.
$$
In the nominal case, we know the following hold:
\begin{align}
x_1 &= Ax_0 + B_uu_0, \nonumber \\
x_2 &= A^2x_0 + AB_uu_0 + B_uu_1, \nonumber \\
&\vdots \nonumber \\
x_N &= A^Nx_0 + A^{N-1}B_uu_0 + \ldots + B_uu_{N-1}. \label{eq:xN_nomsaf}
\end{align}
Let 
$$
\mbf{G}_u =
\begin{bmatrix}
B_u & 0 & \ldots & 0 \\
AB_u & B_u & \ldots & 0 \\
\vdots & \vdots & \ddots & \vdots \\
A^{N-1}B_u & A^{N-2}B_u & \ldots & B_u
\end{bmatrix} \in \mathbb{R}^{nN\times mN}.
$$%
Then, the specification $\mbf{H_{\Gamma}x} \leq \mbf{h_{\Gamma}}$ is rewritten using \eqref{eq:xN_nomsaf} as
\begin{equation} \label{eq:SafetySpec}
\mbf{H_{\Gamma}} \begin{bmatrix} Ax_0 \\ A^2x_0 \\ \vdots \\ A^Nx_0 \end{bmatrix} + \mbf{H_{\Gamma}} \mbf{G}_u\mbf{u} \leq \mbf{h_{\Gamma}}.
%F_{\textsf{saf}} \otimes B_u \mbf{u} 
\end{equation}
Under this constraint, the nominal energy \eqref{eq:EN_def} is obtained from the following quadratic program:
\begin{align}
E_{nom}(x_0; \psi) = &\min_{\mbf{u}} \|\mbf{u}\|_{2}^{2} \nonumber \\
\text{s.t. }~ &\mbf{H_{\Gamma}} \mbf{G}_u \mbf{u} \leq \mbf{h_{\Gamma}} - \mbf{H_{\Gamma}} \begin{bmatrix} Ax_0 \\ A^2x_0 \\ \vdots \\ A^Nx_0 \end{bmatrix}, \nonumber \\
&H_{\U}u_t \leq h_{\U} \text{ for all $t$}. \label{eq:EN_saf}
\end{align}
Next, in the malfunctioning case, we know
\begin{align}
x_1 &= Ax_0 + B_uu_0 + B_ww_0, \nonumber \\
x_2 &= A^2x_0 + AB_uu_0 + B_uu_1 + AB_ww_0 + B_ww_1, \nonumber \\
&\vdots \nonumber \\
x_N &= A^Nx_0 + A^{N-1}B_uu_0 + \ldots + B_uu_{N-1} \nonumber \\
&+ A^{N-1}B_ww_0 + \ldots + B_ww_{N-1}. \label{eq:xN_malfsaf}
\end{align}
Using similar arguments to the nominal case above and replacing $\mbf{w}$ by its worst-case effect as in \eqref{eq:xN_malf2}, the malfunctioning energy \eqref{eq:EM_def} is obtained from the following quadratic program:
\begin{align}
E_{mal}(x_0; \psi) &= \min_{\mbf{u}} \|\mbf{u}\|_{2}^{2} \nonumber \\
\text{s.t. }~ \mbf{H_{\Gamma}} \mbf{G}_u \mbf{u} &\leq \mbf{h_{\Gamma}} - \mbf{H_{\Gamma}} \begin{bmatrix} Ax_0 \\ A^2x_0 \\ \vdots \\ A^Nx_0 \end{bmatrix} - \|\mbf{H_{\Gamma}} \mbf{G}_w \|_1 \overline{w}\mbf{1}_{nN}, \nonumber \\
H_{\U}u_t &\leq h_{\U} \text{ for all $t$}, \label{eq:EM_saf}
\end{align}
where
$$
\mbf{G}_w =
\begin{bmatrix}
B_w & 0 & \ldots & 0 \\
AB_w & B_w & \ldots & 0 \\
\vdots & \vdots & \ddots & \vdots \\
A^{N-1}B_w & A^{N-2}B_w & \ldots & B_w
\end{bmatrix} \in \mathbb{R}^{nN\times pN}.
$$%
Then, the metric \eqref{eq:rA_def} for this specification can be obtained using the difference between \eqref{eq:EM_saf} and \eqref{eq:EN_saf}.

\begin{figure*}[t]
    \centering
    \includegraphics[width=0.9\linewidth]{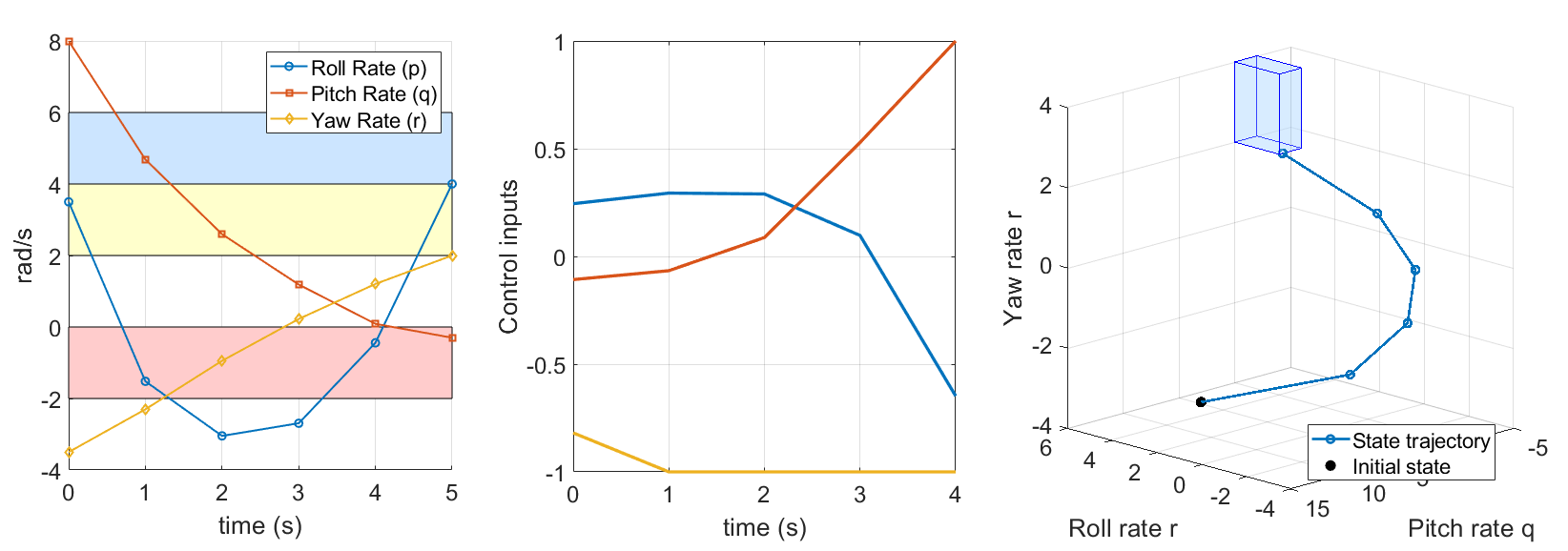}
    \caption{ADMIRE system under the exact-time reachability specification $\nex^5(x_k \in \mathcal{T})$. \emph{Left:} Evolution of the system states over time, showing that the trajectory starts outside the target set and reaches the prescribed target region at the required time instant. \emph{Middle:} Corresponding control inputs used to achieve the reachability task. \emph{Right:} Three-dimensional state evolution of the system, illustrating convergence of the trajectory to the target region at the specified time.}
    \label{fig:admire_reach}
\end{figure*}

\begin{remark}[Composing Specifications]
The results of this section can be used to compute the energetic resilience metric for compositions of these specifications, such as sequential reachability and safety tasks. Obtaining exact bounds on the metric under such compositions is an important direction for future work, along the lines of techniques used in Lemma~\ref{lem:energy_composition_full}. Such a direction would enable addressing a larger class of problems than just those discussed here.
%Computing the metric for such compositions would use results from Lemma~\ref{lem:energy_composition_full}, thus addressing a larger class of problems than just those discussed here.
\end{remark}

\section{Case Studies} \label{sec:Examples}

We validate the proposed energetic resilience framework on two systems: (i) the ADMIRE fighter jet model, and (ii) a planar mobile robot. These case studies highlight the applicability of the framework to safety-critical control, temporal logic specifications, and the dependence of energetic resilience on disturbance magnitude and initial conditions.

\subsection{ADMIRE Aircraft}
In this example, we consider the ADMIRE fighter jet model \cite{SDRA}, a popular benchmark for control applications \cite{KWT10, BO23}. We adopt the linearized, discretized model from \cite{PAOO25}, which consists of the subsystem of the fighter jet model associated with control actions:
\begin{equation}
x_{k+1} = A x_k + B_u u_k + B_w w_k,
\end{equation}
where
$$
A =
\begin{bmatrix}
0.355 & 0 & 0.3428 \\
0 & 0.6031 & 0 \\
-0.0521 & 0 & 0.7901
\end{bmatrix},
$$
% $$
% B_u =
% \begin{bmatrix}
% 0 & -2.72 & 2.72 & 0.7376 \\
% 1.298 & -0.9996 & -0.9996 & 0.0019 \\
% 0 & -0.1153 & 0.1153 & -0.8362
% \end{bmatrix}
% $$
$$
B_u =
\begin{bmatrix}
0 & -2.72 & 2.72 \\
1.298 & -0.9996 & -0.9996 \\
0 & -0.1153 & 0.1153
\end{bmatrix},
\
B_w =
\begin{bmatrix}
0.7376 \\
0.0019 \\
-0.8362
\end{bmatrix}.
$$
The state $x = [p,q,r]^\top \in \mathbb{R}^3$ consists of the roll, pitch and yaw rates measured in rad/s. The control input $u_k \in \R^3$ and $w_k$ is an adversarial input satisfying $\|w_k\|_\infty \leq \overline{w}$. This setting falls under the framework of a partial loss of control authority \cite{BO23, PO26a}.

We consider two separate temporal logic tasks: finite-horizon safety $\always^6(x_k \in \mathcal{S})$ and exact-time reachability $\nex^5(x_k \in \mathcal{T})$, where $\mathcal{S} = \mathcal{T}$ are polytopic sets of the form
\begin{gather*}
    \{x \in \R^3 \mid Gx \leq H\}, \\
G = \begin{bmatrix}
1 & 0 & 0 & -1 & 0 & 0 \\
0 & 1 & 0 & 0 & -1 & 0 \\
0 & 0 & 1 & 0 & 0 & -1
\end{bmatrix}^\top,
\\
H = \begin{bmatrix}
6 & 0 & 4 & -4 & 2 & -2
\end{bmatrix}^\top.    
\end{gather*}

The results in Figures~\ref{fig:admire_safety} and \ref{fig:admire_reach} show that the ADMIRE system successfully satisfies both the finite-horizon safety and exact-time reachability specifications under bounded disturbances. For the disturbance bound $\overline{w} = 0.1$, the computed energetic resilience is $1.787$ for the safety task and $0.048$ for the reachability task, quantifying how much additional energy is required in the malfunctioning case in each example. 
For the safety task, the controller has to keep correcting the effect of disturbances at every time step in order to stay inside the safe set throughout the trajectory. In contrast, for exact-time reachability, the system can tolerate deviations along the way as long as it reaches the target at the required final time, which gives more flexibility and leads to lower additional energy.
Understanding how energetic resilience varies across more complex classes of temporal logic specifications is an interesting direction for future work.

\subsection{Mobile Robot: Sequential Reachability}

We consider a planar mobile robot modeled as a discrete-time integrator:
\begin{equation} \label{eq:mobile}
x_{k+1} = x_k + u_k + w_k,
\end{equation}
where the system state $x \in \mathbb{R}^2$ is the robot position, the control input $u \in \R^2$ is the velocity, and $w$ is a bounded disturbance.

\begin{figure}[t]
    \centering
    \includegraphics[width=0.8\linewidth]{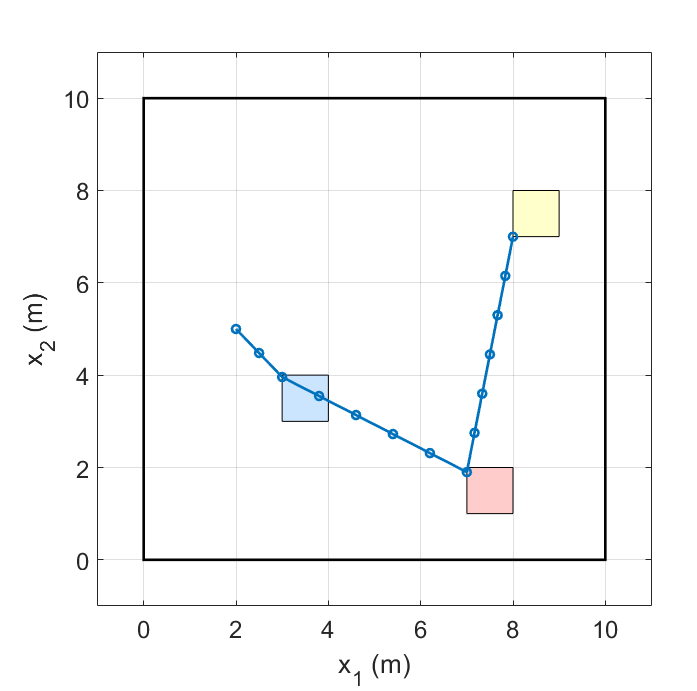}
    \caption{Mobile robot under a sequential reachability task. The robot sequentially reaches $\mathcal{T}_1$ (blue), $\mathcal{T}_2$ (red), and $\mathcal{T}_3$ (yellow), while remaining inside the prescribed safe set $\mathcal{S}$ defined by the black boundary.}
    \label{fig:mobile_seq}
\end{figure}

We consider a task of sequential reachability together with the safety constraint:
\begin{equation}
\eventually^2(\mathcal{T}_1 \wedge \eventually^7(\mathcal{T}_2 \wedge \eventually^{13}(\mathcal{T}_3))) \wedge \always^{13}(x_k \in \mathcal{S}),
\end{equation}
i.e., the system must reach three target sets in sequence while always remaining within a safe region. Here the sets $\mathcal{T}_i, i \in \{1,2,3\}$ and $\mathcal{S}$ are polytopic sets of the form
\begin{gather*}
\{x \in \R^2 \mid Gx \leq H\}, \\
    G = 
\begin{bmatrix}
1 & 0 & -1 & 0 \\
0 & 1 & 0 & -1 
\end{bmatrix}^\top.    
\end{gather*}
and
\begin{align*}
H &= [10 \;\; 10 \;\; 0 \;\; 0]^\top, &&\text{for } \mathcal{S} \\
H &= [4 \;\; 4 \;\; -3 \;\; -3]^\top, &&\text{for } \mathcal{T}_1 \\
H &= [8 \;\; 2 \;\; -7 \;\; -1]^\top, &&\text{for } \mathcal{T}_2 \\
H &= [9 \;\; 8 \;\; -8 \;\; -7]^\top &&\text{for } \mathcal{T}_3.
\end{align*}

This task is more complicated than the basic conjunctions considered in \ref{lem:energy_composition_full}, but this example seeks to show how our framework can extend to specifications whose time horizons may overlap. We compute an upper bound on the energetic resilience by computing the energetic resilience for each component task separately.

The results in Figure~\ref{fig:mobile_seq} show that the mobile robot successfully satisfies the sequential reachability task while remaining within the safe set under bounded disturbances. For the disturbance bound $\overline{w} = 0.01$, the computed upper bound on the energetic resilience is $0.0788$. This upper bound may be conservative since we simply combine the value of the metric for each component task (Lemma~\ref{lem:energy_composition_full}). However, its small value is a consequence of the simplicity of the drift-less dynamics \eqref{eq:mobile} and the small value of $\overline{w}$ considered. In the next subsection, we further investigate how the energetic resilience metric varies with the initial state as well $\overline{w}$.

\begin{figure}[t]
    \centering
    \includegraphics[width=0.8\linewidth]{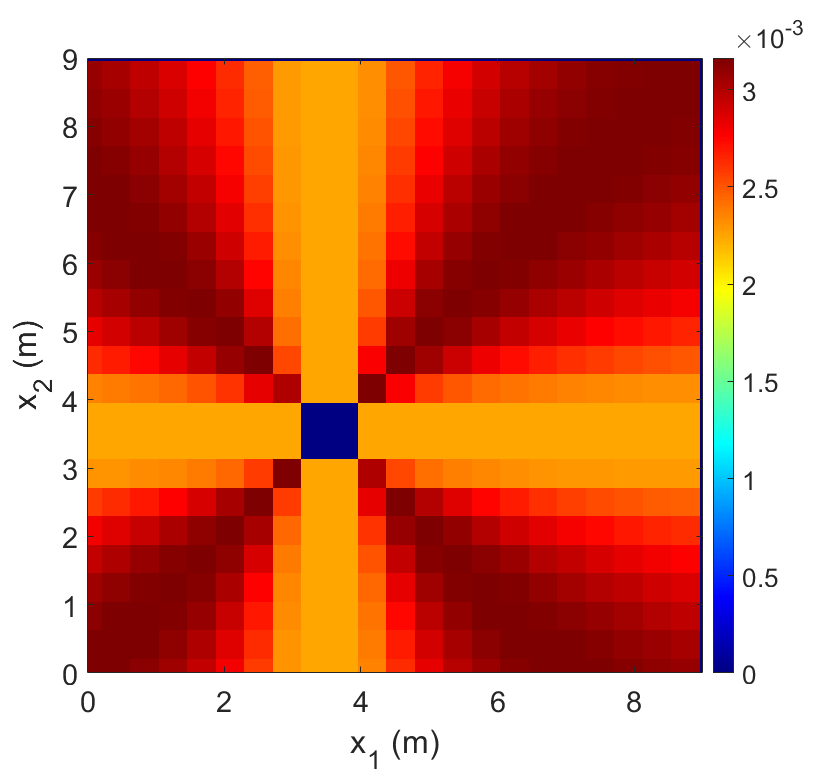}
    \caption{Heatmap of energetic resilience over a grid of initial conditions for the mobile robot. Lower values near the target, moderate values along principle axis of motion, and higher values along diagonal directions highlight initial conditions that are more sensitive to disturbances.}
    \label{fig:hm}
\end{figure}

\begin{figure}[t]
    \centering
    \includegraphics[width=0.8\linewidth]{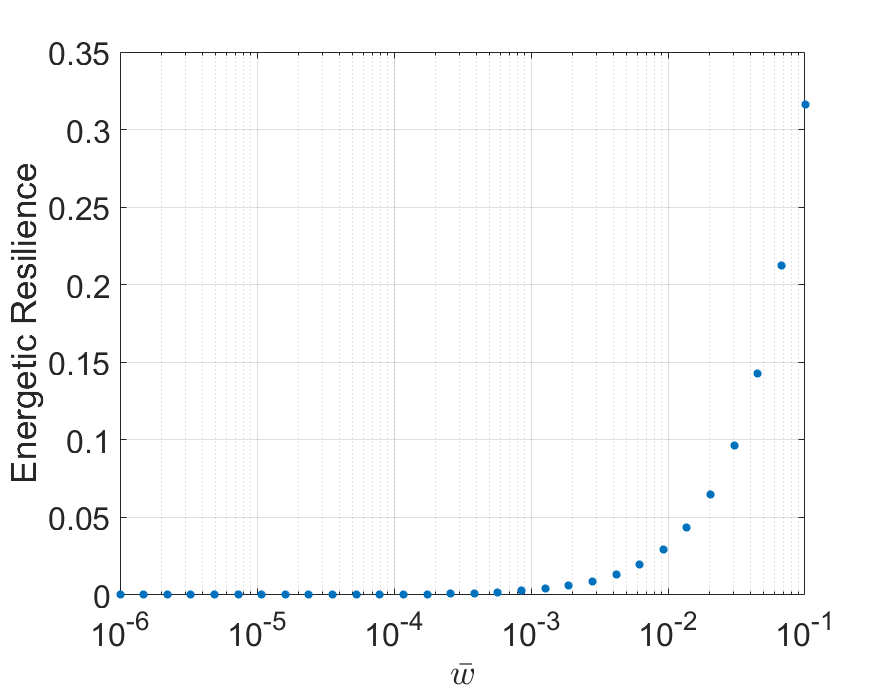}
    \caption{Energetic resilience as a function of disturbance magnitude $\overline{w}$. Resilience increases monotonically with $\overline{w}$, reflecting the additional control effort required to reject larger disturbances.}
    \label{fig:res_vs_w}
\end{figure}

\subsection{Energetic Resilience Analysis}

In this section, we present a brief analysis of the energetic resilience metric to understand its variation with the initial state and magnitude of the disturbance. We consider a simplified single reachability task for the mobile robot system.

\subsubsection{Dependence on Initial State}

We evaluate energetic resilience over a grid of initial conditions. The heatmap in Figure~\ref{fig:hm} shows that resilience is highly state-dependent:
\begin{itemize} \item Resilience is lowest for initial states inside or very close to the target set, since the specification is already satisfied and requires negligible control effort, even under disturbances. As we move away from the target, resilience increases, but not uniformly across directions. \item A “+”-shaped region centered at the target exhibits moderate resilience values, corresponding to directions aligned with the principal directions of the system dynamics and control inputs, which require relatively low energy for reaching under disturbance. \item In contrast, diagonal directions forming a “$\times$”-shaped pattern exhibit significantly higher resilience. These directions are less aligned with the control inputs and require coordinated actuation across multiple inputs, leading to higher control effort and a larger gap between nominal and malfunctioning energies. \end{itemize}
Overall, this reflects the anisotropic controllability of the system, where energetic resilience captures directional variations in control effort for task satisfaction and disturbance rejection.

\subsubsection{Dependence on Disturbance Magnitude}

For a fixed initial condition, we analyze energetic resilience as a function of the disturbance bound $\overline{w}$. The results in Figure~\ref{fig:res_vs_w} indicate a monotonic increase in resilience with $\overline{w}$, while resilience approaches 0 as $\overline{w} \to 0$, and approaches zero as $\overline{w} \to 0$, since the malfunctioning dynamics converge to the nominal case.

This behavior is intuitive: larger disturbance magnitudes require greater control effort to counteract their effect while still satisfying the temporal logic task. As a result, the gap between the malfunctioning and nominal energies increases with $\overline{w}$, leading to higher energetic resilience.

\section{Conclusions} \label{sec:Conclusion}
In this paper, we presented a metric that quantifies the resilience of a control system subjected to undesirable effects, under temporal logic specifications. Such specifications are used to encode complex specifications such as polytopic safety and sequential reachability. Our metric quantified how much additional energy is required to satisfy the given specification under undesired effects, when compared to the system's nominal behavior. We derived some compositional properties of this metric in the context of temporal logic specifications, and demonstrated how computing this metric reduces to solving quadratic programs for certain classes of specifications. Two case studies demonstrated how the proposed framework can be used to synthesize minimum-energy control inputs and how the resilience metric varies based on the initial state and magnitude of undesired effects. Future work will be devoted to exploring more complex system dynamics and undesired effects, as well as more complex specifications.

\balance
\bibliographystyle{IEEEtran}
\bibliography{references}

@BOOK{BK08,
author = {Baier, C. and Katoen, J.-P.},
title = {Principles of Model Checking},
year = {2008},
publisher = {The MIT Press},
address = {Cambridge, MA, USA}
}

@BOOK{Ziegler95,
author = {Ziegler, G. M.},
title = {Lectures on Polytopes},
year = {1995},
publisher = {Springer},
address = {New York, NY, USA}
}

@INPROCEEDINGS{PO25a,
title={Energetic Resilience of Linear Driftless Systems}, 
author={R. Padmanabhan and M. Ornik},
booktitle = {11th IFAC Symposium on Robust Control Design},
month = jul,
year = {2025},
address = {Porto, Portugal}
}

@ARTICLE{PO26a,
author = {R. Padmanabhan and M. Ornik},
title = {Approximate Energetic Resilience of Nonlinear Systems under Partial Loss of Control Authority},
journal = {Automatica},
year = {2026},
volume = {187},
month = may
}

@ARTICLE{SJS25,
author={Saoud, Adnane and Jagtap, Pushpak and Soudjani, Sadegh},
title = {Temporal Logic Resilience for Dynamical Systems},
journal = {IEEE Transactions on Automatic Control},
year = {2025},
note = {early access},
doi = {10.1109/TAC.2025.3626611}
}

@TECHREPORT{SDRA,
author = {Forssell, L. and Nilsson, U.},
title = {{ADMIRE The} aero-data model in a research environment version 4.0, model description},
institution = {FOI -- Swedish Defence Research Agency},
number = {FOI-R--1624--SE},
month = dec,
year = {2005},
url = {https://www.foi.se/rest-api/report/FOI-R--1624--SE}
}

@ARTICLE{PAOO25,
author = {R. Padmanabhan and A. Aspeel and N. Ozay and M. Ornik},
title = {Mode-Prefix-Based Control of Switched Linear Systems with Applications to Fault Tolerance},
journal = {IEEE Control Systems Letters},
volume = {9},
pages = {1784--1789},
year = {2025},
month = jun
}

@ARTICLE{Yan23,
author = {Yan, Y. and Wang, X.-F. and Marshall, B. J. and Liu, C. and Yang, J. and Chen, W.-H},
title = {Surviving disturbances: {A} predictive control framework with guaranteed safety},
journal = {Automatica},
volume = {158},
month = dec,
year = {2023}
}

@ARTICLE{BO23,
author = {Bouvier, J.-B. and Ornik, M.},
title = {Resilience of linear systems to partial loss of control authority},
journal = {Automatica},
volume = {152},
year = {2023},
month = jun
}

@ARTICLE{AH19,
author = {Amin, A. A. and Hasan, K. M.},
title = {A review of Fault Tolerant Control Systems: {Advancements} and applications},
journal = {Measurement},
volume = {143},
month = sep,
year = {2019},
pages = {58--68}
}

@article{rehak19,
  title={Complex approach to assessing resilience of critical infrastructure elements},
  author={Rehak, David and Senovsky, Pavel and Hromada, Martin and Lovecek, Tomas},
  journal={International Journal of Critical Infrastructure Protection},
  volume={25},
  month = jun,
  pages={125--138},
  year={2019}
}

@ARTICLE{PXG20,
  author={Pang, Y. and Xia, H. and Grimble, M. J.},
  journal={IEEE Transactions on Systems, Man, and Cybernetics: Systems}, 
  title={Resilient Nonlinear Control for Attacked Cyber-Physical Systems}, 
  year={2020},
  volume={50},
  number={6},
  pages={2129--2138},
  month = jun
}

@ARTICLE{ZZ24,
author = {Zhao, R. and Zuo, Z. and Tan, Y. and Wang, Y. and Zhang, W.},
title = {Resilient control of networked switched systems subject to deception attack and {DoS} attack},
journal = {Automatica},
volume = {169},
month = nov,
year = {2024}
}

@BOOK{HLXZ21,
author = {Hu, Q. and Li, B. and Xiao, B. and Zhang, Y.},
title = {Control Allocation for Spacecraft Under Actuator Faults},
publisher = {Springer Nature},
address = {Singapore},
year = {2021}
}

@ARTICLE{FMPS19,
  author={Farahani, Samira S. and Majumdar, Rupak and Prabhu, Vinayak S. and Soudjani, Sadegh},
  journal={IEEE Transactions on Automatic Control}, 
  title={Shrinking Horizon Model Predictive Control With Signal Temporal Logic Constraints Under Stochastic Disturbances}, 
  year={2019},
  volume={64},
  number={8},
  pages={3324--3331},
  month = aug}

@INPROCEEDINGS{IHL23,
  author={Ilyes, Roland B. and Ho, Qi Heng and Lahijanian, Morteza},
  booktitle={2023 IEEE International Conference on Robotics and Automation}, 
  title={Stochastic Robustness Interval for Motion Planning with Signal Temporal Logic}, 
  year={2023},
  month = may,
  address = {London, United Kingdom}
  }

@ARTICLE{Liu24,
  author={Liu, Di and Mair, Sebastian and Yang, Kang and Baldi, Simone and Frasca, Paolo and Althoff, Matthias},
  journal={IEEE Transactions on Intelligent Vehicles}, 
  title={Resilience in Platoons of Cooperative Heterogeneous Vehicles: {Self}-Organization Strategies and Provably-Correct Design}, 
  year={2024},
  volume={9},
  number={1},
  month = jan,
  pages={2262-2275},
  }

@INPROCEEDINGS{Zhang24,
  author={Zhang, Sihua and Zhai, Di-Hua and Dai, Xiaobing and Huang, Tzu-Yuan and Xia, Yuanqing and Hirche, Sandra},
  booktitle={2024 IEEE Conference on Decision and Control}, 
  title={Learning-based Parameterized Barrier Function for Safety-Critical Control of Unknown Systems}, 
  year={2024},
  month = dec,
  address = {Milan, Italy}
  }

@INPROCEEDINGS{SJS23,
  author={Saoud, Adnane and Jagtap, Pushpak and Soudjani, Sadegh},
  booktitle={2023 IEEE Conference on Decision and Control}, 
  title={Temporal Logic Resilience for Cyber-Physical Systems}, 
  year={2023},
  month = dec,
  address = {Singapore, Singapore}
  }

@ARTICLE{Na20,
author = {Na, G. and Park, G. and Turri, V. and Johansson, K. H. and Shim, H. and Eun, Y.},
title = {Disturbance observer approach for fuel-efficient heavy-duty vehicle platooning},
journal = {Vehicle System Dynamics},
volume = {58},
number = {5},
month = may,
year = {2020},
pages = {748--767}
}

@ARTICLE{HV1,
author = {Opila, D. F. and Wang, X. and McGee, R. and Gillespie, R. B. and Cook, J. A. and Grizzle, J. W.},
title = {An Energy Management Controller to Optimally Trade Off Fuel Economy and Drivability for Hybrid Vehicles},
journal = {IEEE Transactions on Control Systems Technology},
volume = {20},
number = {6},
month = nov,
year = {2012},
pages = {1490--1505}
}

@ARTICLE{Elias22,
author = {Plaza, E. and Santos, M.},
title = {Management and intelligent control of in-flight fuel distribution in a commercial aircraft},
journal = {Expert Systems},
year = {2022},
note = {early access}
}

@ARTICLE{MS16,
author = {Mersky, A. C. and Samaras, C.},
title = {Fuel economy testing of autonomous vehicles},
journal = {Transportation Research Part C: Emerging Technologies},
volume = {65},
month = apr,
year = {2016},
pages = {31--48}
}

@INPROCEEDINGS{PBDO24,
author = {Padmanabhan, R. and Bakker, C. and Dinkar, S. A. and Ornik, M.},
title = {How Much Reserve Fuel: {Quantifying} the Maximal Energy Cost of System Disturbances},
booktitle = {63rd IEEE Conference on Decision and Control},
address = {Milan, Italy},
month = dec,
year = {2024}
}

@INPROCEEDINGS{KWT10,
  author={Khelassi, Ahmed and Weber, Philippe and Theilliol, Didier},
  booktitle={2010 Conference on Control and Fault-Tolerant Systems (SysTol)}, 
  title={Reconfigurable control design for over-actuated systems based on reliability indicators}, 
  year={2010},
  pages={365--370},
address = {Nice, France},
month = oct
}

@inproceedings{de2013linear,
	title={Linear temporal logic and linear dynamic logic on finite traces},
	author={De Giacomo, G. and Vardi, M. Y.},
	booktitle={IJCAI'13 Proceedings of the Twenty-Third international joint conference on Artificial Intelligence},
	pages={854--860},
	year={2013},
	organization={Association for Computing Machinery}
}

@inproceedings{monir2025computation,
  title={Computation of Feasible Assume-Guarantee Contracts: A Resilience-based Approach},
  author={Monir, Negar and Si, Youssef AIT and Das, Ratnangshu and Jagtap, Pushpak and Saoud, Adnane and Soudjani, Sadegh},
  booktitle={2025 IEEE 64th Conference on Decision and Control (CDC)},
  pages={8122--8129},
  year={2025},
  organization={IEEE}
}

@inproceedings{ait2025maximally,
  title={Maximally resilient controllers under temporal logic specifications},
  author={Ait Si, Youssef and Das, Ratnangshu and Monir, Negar and Soudjani, Sadegh and Jagtap, Pushpak and Saoud, Adnane},
  booktitle={2025 IEEE 64th Conference on Decision and Control (CDC)},
  pages={5034--5040},
  year={2025},
  organization={IEEE}
}

\end{document}